\documentclass[twoside,11pt,final]{entics} 
\usepackage{enticsmacro}
\usepackage{graphicx}
\usepackage[all]{xy}
\usepackage{blkarray, bigstrut}

\usepackage{graphicx} 

\usepackage{bussproofs}

\usepackage{amsmath, amssymb}
\usepackage{mathtools}
\usepackage{enumitem}

\usepackage{amsfonts} 
\usepackage{quiver}

\usepackage{xcolor}
\usepackage{soul}

\def\conf{MFPS 2025} 	
\volume{5}			
\def\volu{5}			
\def\papno{15}			
\def\mydoi{16876}
\def\lastname{Mio} 
\def\runtitle{Compact Quantitative Theories
of Convex Algebras} 
\def\copynames{M. Mio}    
\begin{document}
\begin{frontmatter}
            \title{Compact Quantitative Theories\\
  of Convex Algebras} 		
  \author{Matteo Mio\thanksref{a}\thanksref{myemail}}	
   \address[a]{CNRS \& ENS de Lyon, France}  							
   \thanks[myemail]{Email: \href{mailto:matteo.mio@ens-lyon.fr} {\texttt{\normalshape
   matteo.mio@ens-lyon.fr}}. Research supported by the ANR project \href{https://anr.fr/Projet-ANR-20-CE48-0005}{\texttt{QuaReMe}} (ANR-20-CE48-0005). } 
\begin{abstract} 
  We introduce the concept of compact quantitative equational theory.
  A quantitative equational theory is defined to be compact if all its consequences are derivable by means of finite proofs.
 We prove that the theory of interpolative barycentric (also known as convex) quantitative algebras of Mardare et.~al.~is  compact. 
 This serves as a paradigmatic example, used to obtain other compact quantitative equational theories of convex algebras, each axiomatizing some distance on finitely supported probability distributions.  
\end{abstract}
\begin{keyword}
  Quantitative Algebra, probability distributions, convex algebras, finite proofs. 
\end{keyword}
\end{frontmatter}
\section{Introduction}\label{introduction:section}

At the core of universal algebra is a sound and complete deductive system, due to 
Birkhoff (see, for example,~\cite[\S 3.2]{DBLP:series/eatcs/Wechler92}),
which allows to derive judgments of the form $E\vdash s=t$, where $E$ is a set of equations, 
with the meaning `every algebra that satisfy all equations in $E$ also satisfies the equation $s=t$'. 
Since Birkhoff's proof system can be seen as a fragment of the deductive apparatus of classical first order logic,
every proof is a finite (both in width and depth) tree.
This is a key property, and it is fundamental when it comes to applications of universal algebra in computer science:%
~proofs, being finite objects, can be represented and manipulated by machines, they can be enumerated and there are effective procedures to check if a given derivation tree is a valid proof or not.


Quantitative algebra, recently introduced in \cite{DBLP:conf/lics/MardarePP16}, is an active area of investigation 
\cite{adamek22,adv23,DBLP:journals/lmcs/BacciBLM18,bmpp18,DBLP:journals/corr/abs-2302-01224,DBLP:journals/lmcs/BacciMPP24,Gavazzo18,DBLP:journals/pacmpl/GavazzoF23,DBLP:conf/lics/MardarePP17,Sarkis2024,MioEA24,MioEA22,DBLP:conf/lics/MioSV21,DBLP:conf/concur/MioV20,DBLP:conf/fscd/LagoHLP22,JurkaEA24,MiliusUrbat19,WildS20,WildS21,GoncharovEA23,DAngeloEA24,ForsterEA25}
aimed at extending
the methods of universal algebra to the study of so-called \emph{quantitative algebras}: algebraic 
structures $\mathbb{A}=(A,\{ op^{\mathbb{A}} \}_{op\in\Sigma})$ further endowed 
with an extended metric distance $d_A:A\times A\rightarrow [0,\infty]$ compatible with the interpretations  $op^{\mathbb{A}}$ of the operations $op\in\Sigma$.
The key novel concept is that of quantitative equation $s=_\epsilon t$ between terms, 
for $\epsilon \in [0,\infty)$,
which asserts that the distance between $s$ and $t$ is bounded above by $\epsilon$ ($d_A(s,t)\leq \epsilon$).

In~\cite{DBLP:conf/lics/MardarePP16}, a sound and complete proof system, analogous to that
of Birkhoff, has been investigated for deriving judgments\footnote{
This is a simplification sufficient for the purpose of this introduction.
More precisely, the judgments of the proof system are called \emph{quantitative inferences}~\cite{DBLP:conf/lics/MardarePP16} and are Horn implications involving quantitative equations.
Precise definitions are given in Section \ref{sec:background:qa}.}
of the form $E\vdash s=_\epsilon t$ 
that involve quantitative equations
\emph{in lieu}
of equations. However, unlike Birkhoff's system, the
deductive system for quantitative algebras is based on the infinitary first order logic
$L_{\omega_1,\omega}$ \cite{hodges93book}
which, crucially, allows countable conjuctions. 
As a consequence, a proof in quantitative algebra is, in general, not a finite object: it is a well founded tree with countable width. 
More specifically, the key source of infinite width in proofs comes from the presence of the following rule:
	\begin{center}
    $
	\text{
	\AxiomC{$E\vdash s=_{\epsilon_0} t$}
	\AxiomC{$E\vdash s=_{\epsilon_1} t$}
	\AxiomC{$\dots$}
	\AxiomC{$E\vdash s=_{\epsilon_n} t$}
	\AxiomC{$\dots$}
	\RightLabel{$\epsilon = \inf \{ \epsilon_i \}_{i\in \mathbb{N}}$}
	\QuinaryInfC{$E\vdash s=_\epsilon t$}
	\DisplayProof}
	$
  \end{center}
which states that if the distance between $s$ and $t$ is bounded above by $\epsilon_i$ (i.e.~$d_A(s,t)\leq \epsilon_i$), for all $i\in\mathbb{N}$, 
then necessarily the distance is also bounded by the infimum of all bounds ($d_A(s,t)\leq \epsilon$).

The non-finiteness of proofs in quantitative equational logic is an obstacle to mechanization, of course.
This is unavoidable: it can be shown
that the presence of an infinitary rule is required in any sound and complete proof system for quantitative algebra. 

\paragraph{Compact Quantitative Theories}
On the other hand, it is still possible 
for some well-behaved quantitative theories $E$, to satisfy the following property
(which is stated precisely in Section \ref{section:compact_theory}):
\begin{center}
If $E\vdash s=_\epsilon t$ is derivable, then it is derivable by a finite proof not using the infinitary rule.
\end{center}
We refer to such quantitative equational theories $E$ as \emph{compact}.

Perhaps surprisingly, several of the main examples of quantitative equational theories that have appeared in the literature turn out to be compact although,
to the best of our knowledge, this has never been explicitly observed. One such example, which we fully investigate and generalize in this work, is given by the quantitative theory of \emph{interpolative barycentric algebras} of \cite[\S 10]{DBLP:conf/lics/MardarePP16} 
which axiomatizes the Kantorovich (also known as 1-Wasserstein \cite[Ch 6]{villani2008}) metric on finitely supported probability distributions on a metric space.
Other examples include the theory of quantitative semilattices \cite[\S 10]{DBLP:conf/lics/MardarePP16}, axiomatizing the Hausdorff distance on finite subsets of a metric space, and the theory of quantitative convex semilattices \cite{DBLP:conf/concur/MioV20}, axiomatizing the composition of the Kantorovich and Hausdorff distances, on finitely generated convex sets of finitely supported probability distributions on a metric space.


Given the presence of such examples, important in the theory of programming languages and semantics,
and the relevance of finitary proofs in practical applications of  quantitative algebra, we argue that the systematic study of quantitative compact theories is an interesting endeavour.
Natural types of questions which we will  
explore in forthcoming work, include: 
what are conditions on theories that guarantee compactness? For example syntactical conditions on the shape of the quantitative equations,
or categorical conditions on the corresponding term monads. What is common between the examples mentioned above?
 In the other direction, assuming that a theory is compact, what can generally be said about its class of models, the corresponding term monad, \emph{et cetera}?

\paragraph{Contribution and Organization of this Work} The goal of this paper is to set the foundation for this line of research
by formally introducing the notion of compact quantitative theory and proving, in full details, that the theory of \emph{interpolative barycentric algebras} is compact. 
By generalizing this example, we also derive a family of compact quantitative theories, each axiomatising some distance (lifting) on probability distributions,
such as the  $k$-Wasserstein distance for  $k\in[1,+\infty]$, (also studied in \cite{DBLP:conf/lics/MardarePP16} for $k<\infty$) and examples involving log-probabilities.

Rather than working with the original theory of quantitative algebra of \cite{DBLP:conf/lics/MardarePP16}, we formulate our results in the recently
introduced generalization of \cite{MioEA24} (see \S 9.1 of \cite{MioEA24} for a detailed comparison with \cite{DBLP:conf/lics/MardarePP16}). One axis of generalisation consists in allowing quantitative algebras $\mathbb{A}=(A,\{ op^{\mathbb{A}} \}_{op\in\Sigma}, d_A)$
whose distances $d_A$ are not required to be a metric and can be arbitrary fuzzy relations. This, beside allowing non-metric distances to be modelled (as in \cite{MioEA22}),
has the advantage of decoupling the technicalities related to compactness from those related to metric reasoning. 
For example, by casting the quantitative theory of \emph{interpolative barycentric algebras} of \cite{DBLP:conf/lics/MardarePP16} in the context of \cite{MioEA24}, 
as we do in Section \ref{main:section:result}, we obtain a generalized compact axiomatization of the Kantorovich lifting of distances that are not necessarily metrics.

In Section \ref{technical:background:section} we introduce the necessary technical background.
In Section \ref{section:compact_theory} we introduce the notion of compact quantitative theory and provide some examples.
In Section \ref{main:section:result} we prove the compactness of a generalization of the quantitative theory of \emph{interpolative barycentric algebras} of \cite{DBLP:conf/lics/MardarePP16}.
In Section \ref{section:generalisation} we generalize this result and obtain other examples of compact quantitative theories of convex algebras.

\section{Technical Background}\label{technical:background:section}
\subsection{Convex Algebras, Probability Distributions and Couplings}
Convex algebras, also known under serveral other names including convex spaces \cite{fritz2015convexspacesidefinition} and barycentric algebras \cite{Stone1949} (see \cite{lmcs:4036,fritz2015convexspacesidefinition,Jacobs2010} for overviews), 
are algebraic structures over the uncountable signature of operations $\Sigma_\textnormal{CA}=\{ +_p \mid 0 < p < 1\}$ where, for each $p\in  (0,1)$, the operation $+_p$ is binary. Therefore, a $\Sigma_\textnormal{CA}$-algebra is a relational structure $\mathbb{A}=(A, \{ +^{\mathbb{A}}_p \}_{p\in (0,1)})$ with $+^{\mathbb{A}}_p: A\times A \rightarrow A$, for all $p$.
\begin{definition}\label{convex-algebras-definition}
A \emph{convex algebra} is a $\Sigma_\textnormal{CA}$-algebra $\mathbb{A}$ such that:
\begin{center}
\begin{tabular}{l l l}
	Idempotency: &  $a +^{\mathbb{A}}_p a = a $ & $\forall a\in A, \forall p \in (0,1),$\\
	Skew commutativity.: &  $a +^{\mathbb{A}}_p b = b +^{\mathbb{A}}_{1-p} a$ & $\forall a,b\in A, \forall p \in (0,1),$\\
	Skew associativity: &  $(a +^{\mathbb{A}}_p b) +^{\mathbb{A}}_q c = a +^{\mathbb{A}}_{pq} (b +^{\mathbb{A}}_{\frac{(1-p)q}{1-pq}} c)$ & $\forall a,b,c\in A, \forall p,q \in (0,1).$\\
\end{tabular}
\end{center}
\end{definition}
If $\mathbb{A}$ is clear from the context, we just write $a+_p a'$ in place of $a+^{\mathbb{A}}_p a'$.
Any convex subset of a real vector space, such as $[0,1]$, is a convex algebra with operations defined as $x+_p y = px + (1-p)y$.
But not all convex algebras are of this form. For example, any semilattice $(X,\vee)$ becomes a convex algebra by defining $x+_p y = x\vee y$, for all $p\in(0,1)$.

Alternatively, it is possible to present convex algebras using a signature of $n$-ary operations $\sum^{n}_{i=1} p_i x_i$, where $p_i\in [0,1]$ (thus including $0$ and $1$) and $\sum^{n}_{i=1} p_i =1$.
In this version, a convex algebra $\mathbb{A}$ is a set $A$ with interpretations of type $A^n\rightarrow A$, for all $n$-ary convex combination expressions, subject
to axioms similar to those given above for binary operations. The two approaches are equivalent \cite[Prop. 2.3]{lmcs:4036}, and one can always rewrite $n$-ary convex combinations to $\Sigma_\textnormal{CA}$-terms, recursively, as follows:
\begin{center}
$
\textnormal{(case $p_1=1$) }\sum^{n}_{i=1} p_i x_i  =  x_1 \qquad \qquad
\textnormal{(case $p_1=0$) }\sum^{n}_{i=1} p_i x_i  =  \sum^{n}_{i=2} p_i x_i 
$
\end{center}
\begin{center}
$
  \textnormal{(case $0<p_1<1$) }\sum^{n}_{i=1} p_i x_i =   x_1 +_{p_1} \big( \sum^{n}_{i=2} \frac{p_i}{1-p_1} x_i \big).
$
\end{center}
In what follows we will switch between binary and $n$-ary operations as most convenient.

Finitely supported distributions on a set play an important role in the theory of convex algebra.
\begin{definition}
Given a set $X$, a probability distribution on $X$ is a function $\mu\colon X\rightarrow[0,1]$
such that $\sum_{x\in X}\mu(x)=1$. The support of $\mu$ is the set $supp(\mu)=\{ x \mid \mu(x)>0\}$. 
We say that $\mu$ is \emph{finitely supported} if $supp(\mu)$ is finite. 
We denote with $D(X)$ the set of finitely supported probability distributions on $X$. For $x\in X$, we denote with $\delta_x$ the Dirac probability distribution defined as $\delta_x(y)=1$ if $y=x$ and $0$ otherwise.
\end{definition}

The set $D(X)$ with operations $+_p$ defined, for all $\mu,\nu\in D(X)$, as
$
(\mu +_p \nu)(x) = p\mu(x) + (1-p)\nu(x)
$, is a convex algebra which, with some abuse of notation, we also denote by $D(X)$. It is well known (see, for example, \cite{lmcs:4036}) that
$D(X)$ is the free convex algebra generated by $X$. This means that for any convex algebra $\mathbb{A}$, there is a one to one correspondence
between functions $f:X\rightarrow A$ and convex algebra homomorphisms $\hat{f}:D(X)\rightarrow \mathbb{A}$, such that $f = \hat{f}\circ \eta_X$, 
where $\eta_X:X\rightarrow D(X)$ is defined as $\eta_X(x)=\delta_x$. 

It also means that $D(X)$ can be seen as the quotient of $\textnormal{Terms}_{\Sigma_\textnormal{CA}}(X)$, the absolutely free algebra of $\Sigma_\textnormal{CA}$-terms over $X$, 
by the axioms of convex algebras (Definition \ref{convex-algebras-definition}). For this reason, given $s\in \textnormal{Terms}_{\Sigma_\textnormal{CA}}(X)$, 
we denote with $[s]\in D(X)$ the corresponding probability distribution. Formally, if $s$ is of the form (using $n$-ary notation) $s=\sum^n_{i=1}p_i x_i$ then
$[s]=\sum^{n}_{i=1}p_i \delta_{x_i}$.


We now turn attention to couplings of finitely supported probability distributions.
\begin{definition}
Let $X$ be a set and $\mu,\nu\in D(X)$. A \emph{coupling} $\gamma$ of $\mu$ and $\nu$ is a probability distribution $\gamma$ on $X\times X$ whose marginals
are $\mu$ and $\nu$: for all $x\in X$, $\sum_{x,y\in X} \gamma(x,y) = \mu(x)$, and for all $y\in X$, $\sum_{x,y\in X} \gamma(x,y) = \nu(y)$.
We denote with $\Gamma(\mu,\nu)\subseteq D(X\times X)$ the set of all couplings on $\mu$ and $\nu$.
\end{definition}

Note that $\Gamma(\mu,\nu)$ is always nonempty because the independent product $\mu\times\nu$ of $\mu$ and $\nu$, 
defined as $(\mu\times\nu)(\langle x,x'\rangle) = \mu(x)\cdot\nu(x')$, is a coupling. Furthermore, viewing $\Gamma(\mu,\nu)$ as a (closed and bounded) subset of 
$[0,1]^{n\times n}$ ($n=|supp(\mu) \cup supp(\nu)|$) we observe that $\Gamma(\mu,\nu)$ is compact, by the Heine-Borel theorem.
\begin{lemma}\label{Heine-Borel-lemmma}
	Given any set $X$ and $\mu,\nu\in D(X)$, the set of couplings $\Gamma(\mu,\nu)$ is nonempty and compact. 
\end{lemma}
The following simple property of couplings is important: a convex combination of couplings is a coupling of the convex combinations.

\begin{lemma}\label{lemma:soundness}
Let $X$ be a set, and $\mu_1,\mu_2,\nu_1,\nu_2\in D(X)$. Let $\gamma_1\in \Gamma( \mu_1,\nu_1)$ and $\gamma_2\in \Gamma( \mu_2,\nu_2)$. 
Then, for all $p\in(0,1)$, it holds that $\gamma_1 +_p \gamma_2 \in\Gamma(\mu_1 +_p \mu_2 , \nu_1 +_p\nu_2 )$.
\end{lemma}



\subsection{Quantitative Algebra}\label{sec:background:qa}
Quantitative algebra, as originally introduced in \cite{DBLP:conf/lics/MardarePP16}, deals with algebraic structures 
$\mathbb{A}=(A, \{op^{\mathbb{A}}\}_{op\in\Sigma})$, for some signature of function symbols $\Sigma$, further
endowed with an \emph{extended metric} $d_{A}\colon A\times A \rightarrow \mathbb{R}_{\geq 0}\cup \{\infty\}$ which makes all operations
1-Lipschitz:
$
d_{A}\big(  op^{\mathbb{A}}(a_1,\dots, a_n), op^{\mathbb{A}}(a'_1,\dots, a'_n) \big) \leq \max_{i=1\dots n} d_{A}(a_i,a'_i).
$

Here we instead follow a generalisation of the entire apparatus, recently proposed in \cite{MioEA24}. Quantitative algebras, in the sense of \cite{MioEA24}, are algebraic structures 
$\mathbb{A}=(A, \{op^{\mathbb{A}}\}_{op\in\Sigma})$, for some some signature of functions symbols $\Sigma$, further
endowed with an arbitrary\footnote{
  The choice of valuing the distance in $[0,1]$ is purely motivated by simplicity reasons. 
  All the results of \cite{MioEA24} hold when $[0,1]$ is replaced by any complete lattice, 
  such as $ \mathbb{R}_{\geq 0}\cup \{\infty\}$ or any quantale. See \cite{Sarkis2024}.
  } fuzzy relation $d_{A}\colon A\times A \rightarrow [0,1]$ not subject to any additional constraints such as being a metric or
making all operations 1-Lipschitz. Such constraints can be expressed by quantitative equations (cf.~Example \ref{basic-example-qeq}) in the same way that commutativity of a binary operation
can be expressed by equations in universal algebra. We refer to \cite{MioEA24} for a detailed exposition and, in particular, to  \cite[\S 9.1]{MioEA24} for a comparison with  \cite{DBLP:conf/lics/MardarePP16}.  Here we restrict attention only to the key definitions needed, later in Section \ref{section:compact_theory}, to formalize the concept of \emph{compact quantitative theory}.

\begin{definition}
Given a set $A$, a function of type $d:A\times A\rightarrow[0,1]$ is called a \emph{fuzzy relation} on $A$. We also refer to the pair $(A,d)$
as a fuzzy relation. Given two fuzzy relations $(A,d_A)$ and $(B,d_B)$, a function $f:A\rightarrow B$ is \emph{1-Lipschitz} (or also \emph{nonexpansive})
if $d_B\big( f(a), f(a') \big) \leq d_A(a,a')$, for all $a,a'\in A$.
\end{definition}
\begin{definition}\label{discrete_fuzzy_rel}
Given a set $A$ we denote with $d^{A}_1$, or just $d_1$ if $A$ is clear from the context, the \emph{discrete fuzzy relation} $d_{1}^A: A \times A \rightarrow [0,1]$
 defined as $d_{1}^A(a,a')=1$, for all $a,a'\in A$.
\end{definition}
Note that, for all fuzzy relations $(B,d_B)$, any function $f:A\rightarrow B$ is 1-Lipschitz as a map $f:(A,d^1_A)\rightarrow (B,d_B)$.
Also note that fuzzy relations are not required to be metrics nor pseudometrics.
For example, the discrete fuzzy relation is neither a metric nor a pseudometric because $d(a,a)\neq 0$.

Fuzzy relations on finite sets can be represented by finite matrices. For example:
\begin{center}
$
\textnormal{the fuzzy relation }
\Big( \ A = \{a_1, a_2\} \ , \  d_A =\begin{cases}
  (a_1,a_1)\mapsto 0.5\\
  (a_1,a_2)\mapsto 1\\
  (a_2,a_1)\mapsto 0.3\\
  (a_2,a_2)\mapsto 0\\
\end{cases}
\Big) \textnormal{ , is represented by: }
\tiny{
\begin{blockarray}{c c c}
& a_1 & a_2 \\
\begin{block}{c[cc]}
  a_1 &  0.5 & 1   \\
  a_2 & 0.3 & 0  \\
\end{block}
\end{blockarray}
}. 
$
\end{center}

\begin{definition}[Quantitative Algebra]
Let $\Sigma$ be a possibly infinite signature of function symbols $op\in\Sigma$, each having finite arity $ar(op)\in\mathbb{N}$. A \emph{quantitative algebra} is
a structure $\mathbb{A}=(A, \{ op^{\mathbb{A}}\}_{op\in \Sigma}, d_A)$, where:
\begin{enumerate}
	\item $(A, \{ op^{\mathbb{A}}\}_{op\in \Sigma})$ is an algebra, in the usual sense of universal algebra:  
	a set $A$ with interpretations of the operations: $ op^{\mathbb{A}}: A^{ar(op)}\rightarrow A$.
	\item $(A,d_A)$ is a fuzzy relation. 
\end{enumerate}
Given quantitative algebras $\mathbb{A}$ and $\mathbb{B}$, a homomorphism of quantitative algebras is a map $f:A\rightarrow B$ which is both 1-Lipschitz as $f:(A,d_A)\rightarrow (B,d_B)$ and
a homomorphism of the underlying algebras. 
\end{definition}
Whereas equations in ordinary universal algebra must be satisfied by all interpretations of the variables, in quantitative equations
the set of variables, say $B$, is specified in a context, namely a fuzzy relation $d_B\colon B\times B\rightarrow [0,1]$, and crucially
the quantitative equation must only be satisfied by
1-Lipschitz interpretations. This is made formal by the following definition.
\begin{definition}[Quantitative Equation]\label{def:quantitative:equation}
A \emph{quantitative equation} is an expression of one of the two following forms:
\begin{center}
$
\forall (B,d_B).  s =_\epsilon t \qquad \qquad \qquad \forall (B,d_B).  s = t 
$
\end{center}
where $(B,d_B)$ is a fuzzy relation on a (possibly infinite) set $B$, $s,t\in\textnormal{Terms}_\Sigma(B)$ are $\Sigma$-terms built from the set of generators $B$ 
and $\epsilon$ is a real number in $[0,1]$. 
When $(B,d_B)$ is finite, we often write the quantitative equation using matrix notation as, for example:
\begin{center}
$
\forall (
\tiny{
\begin{blockarray}{c c c}
& b_1 & b_2 \\
\begin{block}{c[cc]}
  b_1 &  1 & \epsilon   \\
  b_2 & 1 & 1  \\
\end{block}
\end{blockarray}
}
). \ s =_{\epsilon} t\ .
$
\end{center}
\end{definition}

\begin{definition}[Satisfiability Relation]
Let $\mathbb{A}$ be a quantitative algebra. We define the \emph{satisfiability relation} as follows:
\begin{center}
\begin{tabular}{l l c l}
$\mathbb{A}\models \forall (B,d_B).  s =_\epsilon t \qquad$ & $ \textnormal{iff} $ & $\qquad d_A\big( \iota(s) , \iota(t) \big) \leq \epsilon $ & $ \textnormal{ for all 1-Lipschitz $\iota:(B,d_B)\rightarrow(A,d_A)$}$\\
$\mathbb{A}\models \forall (B,d_B).  s = t \qquad$ & $ \textnormal{iff} $ & $\qquad  \iota(s) = \iota(t) $ & $ \textnormal{ for all 1-Lipschitz $\iota:(B,d_B)\rightarrow(A,d_A)$}$\\
\end{tabular}
\end{center}
where $\iota(s),\iota(t)\in A$ denote the intepretations of the terms $s,t$, using the extension $\iota\colon\textnormal{Terms}_{\Sigma}(B)\rightarrow \mathbb{A}$ of $\iota$ to terms defined, as usual in universal algebra, by structural induction on terms.
\end{definition}

\begin{remark}
  Following the universal algebra textbook \cite{DBLP:series/eatcs/Wechler92} and \cite{MioEA24}, the ``$\forall$'' symbol in quantitative equations is adopted to remind the universal quantification involved in the satisfiability relation.
  In \cite{Sarkis2024} (see also \cite{FordEA21}) the notation $(B,d_B)\vdash s=_\epsilon t$ is used \emph{in lieu} of $\forall (B,d_B).s=_\epsilon t$. We reserve the usage of the ``$\vdash$'' symbol, later on, for
  the syntactical consequence relation of quantitative algebra.
  \end{remark}

\begin{remark}
  The reader familiar with the original  apparatus of quantitative algebra of \cite{DBLP:conf/lics/MardarePP16} will recognize that quantitative equations $\forall(A,d).s=_\epsilon t$, in the sense of 
  Definition \ref{def:quantitative:equation}, coincide with \emph{basic inferences} of \cite{DBLP:conf/lics/MardarePP16} of the form:
  $ \{ a=_{d(a,a')} a' \mid a,a'\in A\}\vdash s=_\epsilon t$. Also note that, unlike \cite{DBLP:conf/lics/MardarePP16}, we defined two types of quantitative equations: with  equality $(=)$ and with quantitative equality $(=_\epsilon)$.
  This is because the models are not necessarily metric spaces, and so equality $(=)$ cannot be equivalently be expressed as distance zero $(=_0)$. We refer to
  \cite[\S 9.1]{MioEA24} for a detailed discussion. 
  \end{remark}

  \begin{remark}
    The {1-Lipschitz} condition on interpretations 
    is on the whole (potentially infinite) set $B$ and not on the finite subset $B'\subseteq B$ of variables appearing in $s,t$.
    To appreciate this point, consider $(B,d_B)= ([0,1],d_{[0,1]})$ with the standard Euclidean metric. 
    All 1-Lipschitz maps $f\colon(B,d_B)\rightarrow (2,d^2_1)$, from $(B,d_B)$ to the two-element set with the discrete fuzzy relation, are constant.
    However, there are non-constant 1-Lipschitz maps $f\colon(\{0,1\},d_{B'})\rightarrow (2,d^2_1)$ where $d_{B'}$ is the restriction of $d_B$ to $\{0,1\}\subsetneq [0,1]$.
    \end{remark}

Note that if no interpretation $\iota\colon(B,d_B)\rightarrow (A,d_A)$ is 1-Lipschitz, the quantitative equation is trivially satisfied.
Also (see comment after Definition \ref{discrete_fuzzy_rel}) if $d_B$ is the discrete fuzzy relation on $B$, every interpretation $\iota:B\rightarrow A$ is 1-Lipschitz.
Hence the familiar type of equations, quantified over all possible interpretations, can be expressed by taking $d_B=d^B_1$ to be the discrete fuzzy relation.

\begin{example}
  The following quantitative equation expresses that $op\in\Sigma$ is commutative.
 \begin{center}
  $ \forall (
    \tiny{
    \begin{blockarray}{c c c}
    & b_1 & b_2  \\
    \begin{block}{c[cc]}
      b_1 &  1 & 1 \\
      b_2 &  1 & 1 \\
    \end{block}
    \end{blockarray}
    }
    ). \ op(b_1,b_2) = op(b_2,b_1) \ .$
  \end{center}
\end{example}

When $d_B$ is not discrete, the 1-Lipschitz restriction on the intepretations becomes meaningful. 
\begin{example}
  The following quantitative equation expresses that any two points  having distance smaller or equal than $0$ (and thus, $0$) must be equal:
  \begin{center}
   $ \forall (
     \tiny{
     \begin{blockarray}{c c c}
     & b_1 & b_2  \\
     \begin{block}{c[cc]}
       b_1 &  1 & 0 \\
       b_2 &  1 & 1 \\
     \end{block}
     \end{blockarray}
     }
     ). \ b_1 = b_2\ .$
   \end{center}
\end{example}
By collecting infinitely many quantitative equations, with varying values defining $d_B$, it is
possible to express common properties such as the symmetry of the distance (for all $\epsilon$, $d(x,y)\leq\epsilon \Rightarrow d(y,x)\leq \epsilon$),
or the 1-Lipschitz property of an operation $op\in \Sigma$ (for all $\epsilon$, $d(x,y)\leq \epsilon \Rightarrow d(op(x),op(y))\leq \epsilon$).

\begin{example}\label{basic-example-qeq1}
Let $E=\{\phi_\epsilon\}$ be the set of quantitative equations, indexed by $\epsilon\in[0,1]$, of the form:
\begin{center}
  $\phi_\epsilon=\qquad \forall (
    \tiny{
    \begin{blockarray}{c c c}
    & b_1 & b_2  \\
    \begin{block}{c[cc]}
      b_1 &  1 & \epsilon \\
      b_2 &  1 & 1 \\
    \end{block}
    \end{blockarray}
    }
  ). \ b_2 =_\epsilon b_1.$
\end{center}
A quantitative algebra $\mathbb{A}$ satisfies $E$ (i.e., $\forall \phi_\epsilon \in E,\mathbb{A}\models\phi_\epsilon$) if and only if $d_A$ is symmetric. 
\end{example}

\begin{example}\label{basic-example-qeq}
  Let $op\in\Sigma$ be unary and let $E=\{\phi_\epsilon\}$ be the set of quantitative equations, indexed by $\epsilon\in[0,1]$, of the form:
  \begin{center}
    $\phi_\epsilon=\qquad \forall (
      \tiny{
      \begin{blockarray}{c c c}
      & b_1 & b_2  \\
      \begin{block}{c[cc]}
        b_1 &  1 & \epsilon \\
        b_2 &  1 & 1 \\
      \end{block}
      \end{blockarray}
      }
    ). \ op(b_1) =_\epsilon op(b_2).
      $
  \end{center}
  A quantitative algebra $\mathbb{A}$ satisfies $E$ if and only if $op^{\mathbb{A}}: (A,d_A)\rightarrow (A,d_A)$ is 1-Lipschitz. 
  \end{example}

With these definitions in place, one defines as expected the \emph{consequence relation} between a set $E$ of quantitative equations and a quantitative equation $\psi$.

\begin{definition}[Consequence Relation]
Let $E$ be a set of quantitative equations and $\psi$ be a quantitative equation. We write $E\models \psi$, and say that \emph{$\psi$ is a consequence of $E$}, if:
\begin{center}
$
\forall \mathbb{A}, \ \textnormal{ if } \big( \bigwedge_{\phi \in E} \mathbb{A}\models \phi \big) \ \textnormal{ then } \  \mathbb{A}\models \psi.
$
\end{center}
\end{definition}
A key fact in quantitative algebra is that the consequence relation can be axiomatized: there exists a proof system for deriving judgments of the form $E\vdash \psi$
such that $E\vdash \phi$ is derivable if and only if $E\models \phi$ holds. This proof system, introduced in \cite[\S 4]{MioEA24}, is the quantitative algebra analogous of the well known
deductive system of Birkhoff for equations. It includes a number of \emph{finitary} axioms and rules, some for handling equality judgments (entirely analogous to those of Birkhoff's proof system) like:
\begin{center}
$
    \AxiomC{}
    \RightLabel{}
    \UnaryInfC{$E\vdash \forall (B,d_B). s=s $}
    \DisplayProof
 \qquad \qquad
	\AxiomC{$E\vdash \forall (B,d_B). s= t$}
	\RightLabel{}
	\UnaryInfC{$E\vdash \forall (B,d_B). t= s$}
	\DisplayProof
$
\end{center}
and some for handling quantitive equality judgments, like:

\begin{center}
$
    \AxiomC{}
    \RightLabel{$d_B(b,b')= \epsilon$}
    \UnaryInfC{$E\vdash \forall (B,d_B). b=_\epsilon b'$}
    \DisplayProof
 \qquad \qquad
	\AxiomC{$E\vdash \forall (B,d_B). s=_{\epsilon} t$}
	\RightLabel{$\delta \geq \epsilon$}
	\UnaryInfC{$E\vdash \forall (B,d_B). s=_\delta t$}
	\DisplayProof
$
\end{center}
The axiom on the left allows to derive distances (provided by the context $d_B$) between variables and the ``weakening'' rule on the right 
  states that if from $E$ the upper bound $\leq\epsilon$ on the distance between $s$ and $t$ is derivable, 
then also the weaker upper bound $\leq \delta$  is derivable. An important feature of the proof system is the substitution rule (see Rule 4.e in Definition 4.1 of \cite{MioEA24}):
\begin{center}
  $
    \AxiomC{$E\vdash \forall (B,d_B). s=_\star t$}
    \AxiomC{$\big\{ E\vdash \forall (C,d_C). \sigma(b)=_{d_B(b,b')} \sigma(b') \mid  b,b'\in B,\ d_B(b,b')<1 \big\}$ }
    \RightLabel{$\sigma:B\rightarrow \textnormal{Terms}_\Sigma(C)$}
    \BinaryInfC{$E\vdash \forall (C,d_C). \sigma(s) =_\star \sigma(t).$}
    \DisplayProof
$
\end{center}
where $s=_\star t$ is either $s=t$ or $s=_\epsilon t$ for some $\epsilon\in[0,1]$ and $\sigma(s),\sigma(t)$ denote the application
of the substitution $\sigma$ to the terms $s,t$, defined as usual. The set of premises on the right-side witnesses that the substitution $\sigma$ is 1-Lipschitz. For this reason,
we say that the substitution rule can only be applied with \emph{provably} 1-Lipschitz substitutions $\sigma$.
Note that any substitution $\sigma$ is provably 1-Lipschitz when $d_B=d^B_1$ is the discrete distance on $B$ (in this case the right premise is empty).
Also note that the substitution rule is finitary whenever $\{ b,b' \mid d_B(b,b')<1\}$ is finite, as it happens in the special case $d_B=d^B_1$.
We refer to \cite[\S 4]{MioEA24} for a detailed presentation of the other rules of the proof system.
Here we only focus on the fact that, unlike Birkhoff's proof system, the proof system of quantitative algebra
 includes one infinitary rule: 
	\begin{center}
    $
	\text{
	\AxiomC{$E\vdash \forall (B,d_B). s=_{\epsilon_0} t$}
	\AxiomC{$E\vdash \forall (B,d_B). s=_{\epsilon_1} t$}
	\AxiomC{$\dots$}
	\AxiomC{$E\vdash \forall (B,d_B). s=_{\epsilon_n} t$}
	\AxiomC{$\dots$}
	\RightLabel{$\epsilon = \inf \{ \epsilon_i \}_{i\in \mathbb{N}}$}
	\QuinaryInfC{$E\vdash \forall (B,d_B). s=_\epsilon t$}
	\DisplayProof}
	$
  \end{center}
This rule states that if the bound $\leq \epsilon_i$ (in $\forall (B,d_B). s=_{\epsilon_i} t$) is derivable form $E$, for each $\epsilon_i$, 
then also the ``limit'' $\leq\epsilon$ bound ($E\vdash \forall (B,d_B). s=_\epsilon t)$ is derivable.
This rule has appeared in the literature under several names: it is named \emph{order completeness rule} in \cite{MioEA24},  \emph{Archimedean rule} in \cite{DBLP:conf/lics/MardarePP16} and \emph{continuity rule} in \cite{DBLP:journals/lmcs/BacciMPP24}. 
As a consequence of the presence of this rule, proofs in the proof system of quantitative algebra are, in general, 
well-founded infinite trees because some nodes, corresponding
to the application of the infinitary rule, have countably infinite width. The presence of some kind of infinitary rule is necessary,
in the sense that there is no axiomatisation of the consequence relation ($E \models \phi$) given by
a set of finitary Horn implications
in first order logic. This follows from the fact that the category of pseudometric spaces and 1-Lipschitz maps (which is the category of models of a quantitative theory $E$, see Example \ref{example:noncompact:theory} below) is not locally finitely presentable (see \cite[Ex.~2.2]{AdamekEA15} and \cite[\S 5]{Adamek_Rosicky_1994}).

A main result of quantitative algebra is the existence of free algebras. We refer to \cite[\S 5]{MioEA24} for a detailed formulation.
Here we only state the fact that, given a quantitative equational theory $E$ and a fuzzy relation $(A,d_A)$,
the \emph{free $E$-quantitative algebra generated by $(A,d_A)$} is isomorphic to $\big( \textnormal{Terms}_\Sigma(A)/_\equiv , \{ op^F\}_{op\in\Sigma}, d_F \big)$, where:
\begin{enumerate}
  \item the carrier is the set of $\Sigma$-terms over $A$ modulo the equivalence relation $\equiv$ of ``provable equality'' defined as:
    $s\equiv t\Leftrightarrow E\vdash\forall (A,d_A).s=t$,
  \item the interpretation of the operations is defined as: $op^F([s_1]_{\equiv},\dots,[s_n]_{\equiv})=[op(s_1,\dots,s_n)]_{\equiv}$,
  \item the distance is defined by ``provable $\epsilon$-equality'': $d_F([s]_\equiv,[t]_\equiv)\leq \epsilon \Leftrightarrow E\vdash\forall (A,d_A).s=_\epsilon t$.
  \end{enumerate}
It can be shown that definitions (ii--iii) are valid for all choices of representatives $s\in[s]_\equiv$, $t\in[t]_\equiv$.

\section{Compact Quantitative Theories}\label{section:compact_theory}

We now introduce the concept of a compact quantitative equational theory. 
In what follows a possibly infinite signature $\Sigma$ of function symbols, each having finite arity, is fixed. 

\begin{definition}[Compact Quantitative Theory]\label{compact_theory}
Let $E$ be a set of quantitative equations. We say that $E$ is \emph{compact} if, for all quantitative equations $\phi$,
if $E\models\phi$ then there exists a finite proof (i.e., a proof tree never using infinitary rules) of $E\vdash \phi$.
\end{definition}


\begin{remark}
We use the term \emph{theory} as synonym of \emph{set} of quantitative equations. Often, the term \emph{theory} is instead used to indicate a
set of (quantitative) equations closed under the consequence relation. But, in such setting,
every consequence $\phi$ of $E$ trivially admits a finite proof, because an axiom of the proof system (see ``Init'', Def. 4.1 of \cite{MioEA24}) allows to derive $E\vdash \phi$ whenever $\phi\in E$.
Hence, the notion of compactness of Definition \ref{compact_theory} is meaningful for sets $E$ not closed under deducibility.
\end{remark}


The following sequence of propositions illustrate the notion of compact quantitative theory.

\begin{proposition}
Let $\Sigma=\emptyset$ and $E=\emptyset$. Then $E$ is compact.
\end{proposition}
\begin{proof}
First, observe that since there are no operations ($\Sigma=\emptyset$) the models of $E=\emptyset$ are all fuzzy relations $(A,d_A)$.
Now let $\phi$ be a quantitative equation. We need to prove that if $\emptyset \models \phi$ then $\emptyset \vdash \phi$ has a finite derivation. 
The assumption $\emptyset \models \phi$ means that $(A,d_A)\models \phi$, for all fuzzy relations $(A,d_A)$. 

Let us first consider the case $\phi$ is of the form $\forall (B,d_B). b = b'$. 
We claim that the assumption implies that $b=b'$. Indeed if $b\neq b'$, the fuzzy relation $(B,d_B)$ itself 
and the identity interpretation (which is always 1-Lipschitz), witness that $(B,d_B)\not\models \phi$. 
Now that the claim has been established, we observe that $E \vdash \forall (B,d_B). b = b$ can be derived by one of the axioms of the proof system.

Let us now consider the case $\phi$ is of the form $\forall (B,d_B). b =_\epsilon b'$. 
We claim that the assumption implies that $\epsilon \geq d_B(b,b')$. 
Indeed if $\epsilon < d_B(b,b')$, the fuzzy relation $(B,d_B)$ itself 
and the identity interpretation (which is always 1-Lipschitz), witness that $(B,d_B)\not\models \phi$. 
Now that the claim has been established, we observe that $E\vdash\forall (B,d_B). b =_{d_B(b,b')} b'$ can be can be derived by one of the axioms,
and from it we can apply a weakening rule (because $\epsilon  \geq d_B(b,b')$) to derive $E\vdash \forall (B,d_B). b =_\epsilon b'$.
\end{proof}

As we already observed, the semantics of equations $s=t$, in standard universal algebra, can
be expressed by  quantitative equations of the form $\forall (B, d_1). s= t$, using the discrete fuzzy relation
on a set $B$ containing all variables appearing in $s$ and $t$. All proof rules of Birkhoff's system have corresponding finitary rules
in the proof system of quantitative algebra, when dealing with quantitative equations of the form $\forall (B, d_1). s= t$.
In particular, the familiar substitution rule of Birkhoff proof system (see \cite[Rule EL5]{DBLP:series/eatcs/Wechler92}) has a corresponding rule because any substitution $\sigma$ 
is, as already remarked, provably 1-Lipschitz when using the discrete distance $d_1$. This has the following consequence.

\begin{proposition}\label{example-UA-embedding}
Fix an arbitrary $\Sigma$ and let $E=\{\phi_i\}_{i\in I}$ be a collection of quantitative equations $\phi_i$ of the form
$
 \forall (B_i, d_{1}^{B_i}). s_i = t_i$, where $d_1^{B_i}$ is the discrete fuzzy relation on $B_i$.
Then $E$ is compact.
\end{proposition}
\begin{proof}
Let $\mathcal{E}=\{ s_i = t_i\}_{i\in I}$ be the set of equations (in the sense of universal algebra) corresponding to $E$.
Note that the models of $E$ are quantitative algebras $\mathbb{A}=(A,\{op^{\mathbb{A}}\}_{op\in\Sigma},d_A)$ whose underlying
algebra $(A,\{op^{\mathbb{A}}\}_{op\in\Sigma})$ is a model of $\mathcal{E}$ and $d_A$ is an arbitrary fuzzy relation.
Now, consider an arbitrary quantitative equation $\psi$ and assume $E\models \psi$. We need to show that $E\vdash \psi$ has a finite derivation.

Consider first the case that $\psi$ is of the form $\forall (B,d_B). s=t$, for some $s,t\in\textnormal{Terms}_\Sigma(B)$.
The assumption $E\models \psi$ means that all models of $E$ satisfy $\forall (B,d_B). s=t$. 
In particular, for each model $(A,\{op^{\mathbb{A}}\}_{op\in\Sigma})$ of $\mathcal{E}$, the quantitative algebra 
$\mathbb{A}=(A,\{op^{\mathbb{A}}\}_{op\in\Sigma},d^A_0)$ is a model of $E$,  with $d^A_0$ the constant zero fuzzy relation ($d^A_{0}(a,a')=0$).
Hence $\mathbb{A}\models \forall (B,d_B). s=t$. 
Note that any function $f\colon B\rightarrow A$ is 1-Lipschitz as a map $f\colon (B,d_B)\rightarrow (A,d^A_0)$.
Therefore $\mathbb{A}\models \forall (B,d_B). s=t$ means that $\iota(s)=\iota(t)$ holds for all interpretations $\iota$. 
Hence $\iota(s)=\iota(t)$ holds, for any interpretation $\iota$, in any model of $\mathcal{E}$.
By completness 
of Birkhoff's proof system, this implies that $s=t$ is derivable from $\mathcal{E}$ by means of a finite derivation.
By ``copying'' this finite derivation, in the proof system of quantitative algebra, we obtain a corresponding finite proof of $E\vdash\forall (B,d_1^B). s=t$,
where $d_1^B$ is the discrete fuzzy relation on $B$.
Finally, we can derive $E\vdash\forall (B,d_B). s=t$ from $E\vdash \forall (B,d_1^B). s=t$ by applying a substitution rule of the proof system, with the identity substitution ($\sigma(b)=b$) 
of type $\sigma:(B,d_1^B)\rightarrow \textnormal{Terms}_{\Sigma}(B,d_B)$, which is trivially provably 1-Lipschitz.

 Now consider the case that $\psi$ is of the form $\forall (B,d_B). s=_\epsilon t$. 
 The assumption $E\models \psi$ implies that any model $\mathbb{A}$ of $E$ satisfies $\forall (B,d_B). s=_\epsilon t$. But since, among the models of $E$, 
 there are models with the discrete fuzzy relation $d^A_1$, which assigns distance $1$ to the intepretations of $s$ and $t$, we deduce that $\epsilon =1$.
The quantitative equation $E\vdash \forall (B,d_B). s=_1 t$ can be finitely proved by an axiom of the proof system (see axiom (h), in Definition 4.1 of \cite{MioEA24}).
\end{proof}

However, as soon as quantitative theories $E$ including some quantitative equations of the form $\forall (B,d_B). s=_{\epsilon} t$ are considered, it
is very easy to incur in examples that are not compact.

\begin{proposition}\label{example:noncompact:theory}
  Let $\Sigma=\emptyset$ and let $E$ be the theory of pseudometric spaces defined as the union of the following sets of quantitative equations, for all $\epsilon,\epsilon_1,\epsilon_2\in[0,1]$:
\begin{center}
  $
  \forall (
\tiny{
  \begin{blockarray}{c c c}
& b\\
\begin{block}{c[cc]}
  b & 1   \\
\end{block}
\end{blockarray}
}
). \ b =_0 b
\qquad \qquad 
\forall (
\tiny{
\begin{blockarray}{c c c}
& b_1 & b_2 \\
\begin{block}{c[cc]}
  b_1 &  1 & \epsilon   \\
  b_2 & 1 & 1  \\
\end{block}
\end{blockarray}
}
). \ b_2 =_{\epsilon} b_1
\qquad \qquad 
\forall (
\tiny{
\begin{blockarray}{c c c c}
& b_1 & b_2 & b_3\\
\begin{block}{c[ccc]}
  b_1 & 1 & \epsilon_1 & 1  \\
  b_2 & 1 & 1 & \epsilon_2 \\
  b_3 & 1 & 1 & 1\\
\end{block}
\end{blockarray}
}
). \ b_1 =_{\min\{1,\epsilon_1 + \epsilon_2\}} b_3\ .
$
\end{center}
Then the quantitative theory $E$ is not compact.
\end{proposition}
\begin{proof}
It is easy to verify (see \cite[\S 2.3]{MioEA24}) that the models of $E$ are fuzzy relations $(A,d_A)$ such that $d_A$ is a pseudometric: symmetric ($d_A(a,a')=d_A(a',a)$),
self-distance zero ($d_A(a,a)=0$) and satisfying the triangular inequality ($d_A(a,c) \leq d_A(a,b)+ d_A(b,c)$). Consider the fuzzy relation $(B,d_B)$ where $B=\{\underline{0},\underline{0'}\}\cup \{ \underline{\epsilon}\mid 0< \epsilon \leq 1\}$ (the unit interval with ``two zeros'') and $d_B$ is a symmetric and self-distance zero fuzzy relation specified by:
$
d_B(\underline{0},\underline{0'})=1$, $d(\underline{0},\underline{\epsilon})=d(\underline{0'},\underline{\epsilon}) = \epsilon$ and $d_B(\underline{\epsilon},\underline{\delta})=|\epsilon -\delta|
$, for all $\epsilon,\delta \in (0,1]$:
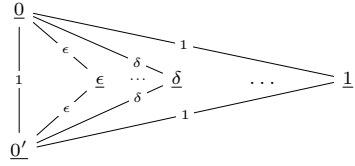
\begin{center}
  \tiny{
\begin{tikzcd}
	\underline{0} \\
	& \underline{\epsilon} & \underline{\delta} & \dots & \underline{1} \\
	{\underline{0'}}
	\arrow["\epsilon"{description, pos=0.6}, no head, from=1-1, to=2-2]
	\arrow["\delta"{description, pos=0.8}, no head, from=1-1, to=2-3]
	\arrow["1"{description}, no head, from=1-1, to=2-5]
	\arrow["\dots"{description}, draw=none, from=2-2, to=2-3]
	\arrow["1"{description}, no head, from=3-1, to=1-1]
	\arrow["\epsilon"{description, pos=0.6}, no head, from=3-1, to=2-2]
	\arrow["\delta"{description, pos=0.8}, no head, from=3-1, to=2-3]
	\arrow["1"{description}, no head, from=3-1, to=2-5]
\end{tikzcd}
}
\end{center}
Hence $d_B$ does not satisfy the triangular inequality because $d_B(\underline{0},\underline{0'})=d_B(\underline{0'},\underline{0})=1$.
Let $\phi$ be the quantitative equation $\forall (B,d_B). \underline{0} =_0 \underline{0'}$. 
We show that $E\models \phi$ but all proofs of $E\vdash \phi$ are infinite (i.e., they involve the infinitary rule).

First, we prove that $E\models \phi$. Let $(A,d_A)$ be a model of $E$, i.e., a pseudometric space, and let $\iota\colon(B,d_B)\rightarrow(A,d_A)$ be a 1-Lipschitz interpretation.
Note that $d_A(\iota(\underline{0}), \iota(\underline{0'}))=0$, because for every $\epsilon>0$, 
$d_A(\iota(\underline{0}), \iota(\underline{0'}))\leq d_A(\iota(\underline{0}), \iota(\underline{\epsilon})) + d_A(\iota(\underline{\epsilon}), \iota(\underline{0'}))\leq \epsilon + \epsilon$, 
by triangular inequality of $d_A$ and the 1-Lipschitz property of $\iota$. As this holds for all 1-Lipschitz $\iota$, we conclude that $(A,d_A)\models \phi$. Hence $E\models \phi$.

Secondly, we show that $E\vdash \phi$ does not have finite proofs. Informally, each judgment  $\forall (B,d_B). \underline{0} =_{2\epsilon} \underline{0'}$ has a finite derivation from the judgments $\forall (B,d_B). \underline{0} =_\epsilon \underline{\epsilon}$ and  $\forall (B,d_B). \underline{\epsilon} =_\epsilon \underline{0'}$, which are
 finitely derivable using an axiom of the proof system, and by application of the ``triangular inequality'' quantitative equation from $E$. But collecting these judgments to obtain $\forall (B,d_B). \underline{0} =_0 \underline{0'}$ requires an application of  the infinitary rule. 
Formally, we need to exhibit a relational structure $(C,\{R_\epsilon\}_{\epsilon\in[0,1]})$,
where each $R_\epsilon\subseteq C\times C$ is a binary relation interpreting $=_\epsilon$, such that $C$ models all the deductive rules of the proof system  -- except the infinitary rule -- and all the quantitative equations in $E$ but not $\phi$.
Let $C=B$ and define $R_0=\textnormal{Id}_C$ (the identity relation) and $R_\epsilon$, for each $\epsilon >0$, as the symmetric reflexive relation satisfying: 
\begin{center}
$
(\underline{0}, \underline{0'})\in R_\epsilon,
\ \ \ 
(z,\underline{\delta})\in R_\epsilon \, (\textnormal{for $z\in\{\underline{0},\underline{0'}\}$ and  $0< \delta\leq \epsilon$}),
\ \ \ 
(\underline{\delta},\underline{\lambda})\in R_\epsilon  \, (\textnormal{for }0<\delta,\lambda\leq 1 \textnormal{ s.t. }|\delta -\lambda|\leq \epsilon).
$
\end{center}
In this model, $\underline{0}$ is arbitrarily close to $\underline{0'}$ (i.e., $(\underline{0}, \underline{0'})\in R_\epsilon$) but not at distance $0$ (i.e., $(\underline{0}, \underline{0'})\not\in R_0$).
\end{proof}

\section{The Theory of Interpolative Convex algebras is compact}\label{main:section:result}

The theory of \emph{Interpolative Barycentric (IB) quantitative algebras} was introduced in \cite{DBLP:conf/lics/MardarePP16} as a main example of quantitative theory.
In \cite{DBLP:conf/lics/MardarePP16} it was proved that the free IB quantitative algebra generated by a metric space $(X,d\colon X\times X\rightarrow [0,1])$ is isomorphic to the quantitative 
algebra $(D(X), \{ +_p\}_{p\in (0,1)}, K(d))$, where $K(d)\colon D(X)\times D(X)\rightarrow [0,1]$ is the Kantorovich (see Definition \ref{Katorovich-defintion} below) lifting of $d$. 
In this section we recast this result in our chosen apparatus of quantitative algebra of \cite{MioEA24}, introduced in Section \ref{sec:background:qa}, where distances are not necessarily metrics. 
To avoid confusion with \cite{DBLP:conf/lics/MardarePP16} we will refer to it as the quantitative theory of \emph{Interpolative convex algebras} ($\mathbb{ICA}$).

\begin{definition}[Theory $\mathbb{ICA}$]\label{definition_thery_ICA}
  Let $\Sigma_\textnormal{CA}$ be the signature of convex algebras. The quantitative theory of \emph{interpolative convex algebras}, denoted by $\mathbb{ICA}$, is defined as the union
  of the following quantitative equations, for all $p,q\in(0,1)$:
  \begin{center}
    \begin{tabular}{l l}
      Idempotency: &  $\forall (
        \tiny{
          \begin{blockarray}{c c c}
        & x\\
        \begin{block}{c[cc]}
          x & 1   \\
        \end{block}
        \end{blockarray}
        }
      ).\ x +_p x = x\qquad$
      \end{tabular}
      \begin{tabular}{l l}
      Skew comm.: &  $\forall (
        \tiny{
        \begin{blockarray}{c c c}
        & x & y \\
        \begin{block}{c[cc]}
          x &  1 & 1   \\
          y & 1 & 1  \\
        \end{block}
        \end{blockarray}
        }
        ).\ x +_p y = y +_{1-p} x$
      \end{tabular}
  \end{center}
\begin{center}
        \begin{tabular}{l l}
      Skew assoc.: &  $\forall (
        \tiny{
        \begin{blockarray}{c c c c}
        & x & y & z\\
        \begin{block}{c[ccc]}
          x & 1 & 1 & 1  \\
          y & 1 & 1 & 1 \\
          z & 1 & 1 & 1\\
        \end{block}
        \end{blockarray}
        }
        ). \ (x +_p y) +_q z = x +_{pq} (y +_{\frac{(1-p)q}{1-pq}} z)$ 
    \end{tabular}
  \end{center}

and the following quantitative equations, for all $p\in(0,1)$ and $\epsilon,\delta\in[0,1]$:
\begin{center}
  \begin{tabular}{l l}
    Interpolative: &  $\forall (
      \tiny{
      \begin{blockarray}{c c c c c}
      & x & y & w & z\\
      \begin{block}{c[cccc]}
        x & 1 & 1 & \epsilon & 1 \\
        y & 1 & 1 & 1 & \delta \\
        w & 1 & 1 & 1 & 1 \\
        z & 1 & 1 & 1 & 1 \\
      \end{block}
      \end{blockarray}
      }
      ). \ x +_p y\  =_{p\epsilon + (1-p)\delta} \ w+_{p} z$.
  \end{tabular}
  \end{center}
  
\end{definition}
The first set of quantitative equations simply corresponds to the equations of convex algebras (Definition \ref{convex-algebras-definition}), which must hold for all possible interpretations,
and are therefore expressed with the discrete fuzzy relation on variables. This implies that all models $\mathbb{A}=(A,\{+^{\mathbb{A}}\}_{p},d_A)$
satisfying $\mathbb{ICA}$ are such that the underlying algebra $(A,\{+^{\mathbb{A}}\}_{p})$ is a convex algebra. The set of ``interpolative'' quantitative equations, on the other hand, have a nontrivial fuzzy relation on variables, which constrains interpretations
of $x$ and $w$ (respectively, $y$ and $z$) to have distance bounded by $\epsilon$ (respectively, bounded by $\delta$).

\begin{definition}[Kantorovich distance lifting]\label{Katorovich-defintion}
  Let $(A,d)$ be a fuzzy relation. The \emph{Kantorovich lifting of $d$} is a fuzzy relation on $D(A)$ (that is $K(d):D(A)\times D(A)\rightarrow [0,1]$) defined as follows, for all $\mu,\nu\in D(A)$:
\begin{center}
 $ K(d)(\mu,\nu) = \displaystyle \inf_{\gamma\in\Gamma(\mu,\nu)} V_{d}(\gamma) \qquad \qquad \textnormal{where}\qquad V_{d}(\gamma)= \displaystyle \sum_{a,b\in A} \gamma(a,b)\cdot d(a,b).$
\end{center}
  \end{definition}
  \begin{remark}
 Since $\Gamma(\mu,\nu)$ is nonempty $K(d)(\mu,\nu)<\infty$. Furthermore, since all values $\gamma(a,b)$ and $d(a,b)$ are in $[0,1]$, we have that $0\leq K(d)(\mu,\nu) \leq 1$. In 
  other words, $K(d)$ is indeed a fuzzy relation on $D(A)$.
  \end{remark}
  \begin{remark}
  It is well known (see e.g.~\cite[p. 68]{villani2008}) that when $d$ is a (pseudo)metric on $A$ then also $K(d)$ is a (pseudo)metric on $D(A)$. In general, however, $K(d)$ does not satisfy these properties.
  For example, if $d(a,a)=1$ for some $a\in A$, then $K(d)(\delta_a,\delta_a)=1$ thus not satisfying the properties of metrics.
  \end{remark}
\begin{proposition}\label{unit_is_nonexpansive}
  Let $(A,d)$ be a fuzzy relation. The map $\eta_A:(A,d)\rightarrow (D(A),K(d))$, defined as $\eta_A(a)=\delta_a$, is 1-Lipschitz.
  \end{proposition}
  \begin{proof}
    We need to show that $K(d)(\delta_a,\delta_b)\leq d(a,b)$, for all $a,b\in A$. This inequality is witnessed by the independent product $\delta_a\times\delta_b$, 
    which is a coupling, and satisfies $V_d(\delta_a\times\delta_b)=d(a,b)$.
  \end{proof}
  The first result, relating the quantitative theory $\mathbb{ICA}$ and the Kantorovich distance, states that the 
  Kantorovich distance gives a model of $\mathbb{ICA}$, a statement that can be rephrased as a soundness property.

\begin{proposition}\label{soundness:proposition1}
Let $(A,d_A)$ be a fuzzy relation and $(D(A),\{+_p\}_{p\in(0,1)})$ be the convex algebra of finitely supported probability distributions.
The quantitative algebra $(D(A),\{+_p\}_{p\in(0,1)}, K(d))$ is a model of $\mathbb{ICA}$. 
\end{proposition}
\begin{proof}
  We know that $D(A)$ satisfies all axioms (idempotency and skew comm./assoc.) of convex algebras, and thus it satisfies
  the corresponding quantitative equations in $\mathbb{ICA}$.
  It remains only to show that it
  satisfies all instances, for $p\in(0,1)$ and $\epsilon,\delta\in[0,1]$, of the interpolative quantitative equation.
  So assume $\iota:\{x,y,w,z\}\rightarrow D(A)$ is any 1-Lipschitz interpretation of the variables and denote with $\mu_x=\iota(x)$,
  $\mu_y=\iota(y)$, $\mu_w=\iota(w)$ and $\mu_z=\iota(z)$. The 1-Lipschitz property of $\iota$ 
  amounts to say that $K(d)(\mu_x, \mu_w)\leq \epsilon$ and $K(d)(\mu_y, \mu_z)\leq \delta$. We need to prove that: $K(d)\big( \mu_x +_p \mu_y,\mu_w +_p \mu_z\big)\leq p\epsilon + (1-p)\delta$.

  By definition of $K(d)$ as an infimum, for any $\lambda>0$, we can find couplings $\gamma\in \Gamma(\mu_x,\mu_w)$ and $\gamma' \in \Gamma(\mu_y,\mu_z)$ such that:
  $
V_d(\gamma) \leq \epsilon +\lambda$
and $V_d(\gamma')  \leq \delta  +\lambda$. From $\gamma$ and $\gamma'$, we obtain by Lemma \ref{lemma:soundness} a coupling $\gamma+_p \gamma' \in \Gamma(\mu_x +_p \mu_y,\mu_w +_p \mu_z)$ such that:
\begin{center}
$
V_d(\gamma+_p \gamma') \stackrel{def}{=} \displaystyle \sum_{a,b\in A} \big(p\gamma(a,b) + (1-p)\gamma'(a,b)\big)\cdot d(a,b) = pV_d(\gamma) + (1-p)V_d(\gamma'). 
$
\end{center}
Hence $V_d(\gamma+_p \gamma') \leq p(\epsilon+\lambda) +(1-p)(\delta + \lambda)$. Since $\lambda>0$ is arbitrary, by taking $\lambda\rightarrow 0$ we deduce that 
$V_d(\gamma+_p \gamma') \leq p\epsilon +(1-p)\delta$. Hence $K(d)\big( \mu_x +_p \mu_y,\mu_w +_p \mu_z\big)\leq p\epsilon + (1-p)\delta$.
\end{proof}

Recall that given $s\in\textnormal{Terms}_{\Sigma_{\textnormal{CA}}}(A)$ a convex algebra term with variables in $A$ we write $[s]\in D(A)$ for 
the corresponding probability distribution, defined as: $[\sum^n_{i=1} p_i a_i] = \sum^n_{i=1} p_i \delta_{a_i}$.

\begin{corollary}[Soundness]\label{soundness:corollary}
  Let $(A,d_A)$ be a fuzzy relation and $s,t\in\textnormal{Terms}_{\Sigma_{\textnormal{CA}}}(A)$ be two convex algebra terms with variables in $A$.
 Then the following implication holds,  for all $\epsilon\in[0,1]$:
  \begin{center}
 $\mathbb{ICA}\vdash \forall (A,d_A).s=_\epsilon t \ \textnormal{ is derivable }\qquad \Longrightarrow \qquad K(d_A)([s],[t])\leq \epsilon$.
  \end{center}
\end{corollary}
\begin{proof}
By the properties of the proof system of quantitative algebra, $\mathbb{ICA}\vdash \forall (A,d_A).s=_\epsilon t$ is equivalent to  $\mathbb{ICA}\models \forall (A,d_A).s=_\epsilon t$.
By Proposition \ref{soundness:proposition1}, $(D(A),\{+_p\}_{p\in(0,1)}, K(d))$ is a model of $\mathbb{ICA}$.
The assumption therefore implies that $(D(A),\{+_p\}_{p\in(0,1)}, K(d))\models  (A,d_A).s=_\epsilon t$.
This means that for all 1-Lipschitz interpretations $\iota\colon A\rightarrow D(A)$ it holds that $K(d_A)(\iota(s),\iota(t))\leq \epsilon$. 
The interpretation $\eta_A$ ($a\mapsto \delta_a$) is 1-Lipschitz by Lemma \ref{unit_is_nonexpansive} and, by definition,
$\eta_A(s)=[s]$ and $\eta_A(t)=[t]$. Hence $K(d_A)([s],[t])\leq \epsilon$ holds.
\end{proof}

The following useful corollary can be also be deduced from Proposition \ref{soundness:proposition1} and the fact
that $D(A)$ is the free convex algebra on $A$. It states, 
using the termilogy established in \cite[Def 7.6]{MioEA24}, that $\mathbb{ICA}$ is a quantitative \emph{extension} of the theory of convex algebras.

\begin{corollary}\label{technical_corollary_proofsystem_extension}
  Let $(A,d_A)$ be a (possibly infinite) fuzzy relation and $s,t\in\textnormal{Terms}_{\Sigma_{\textnormal{CA}}}(A)$ be two convex algebra terms with variables in $A$.
The following are equivalent:
\begin{enumerate}
  \item $s=t$ is provable in Birkhoff's proof system from the axioms of convex algebras (Definition \ref{convex-algebras-definition}). 
\item $\mathbb{ICA}\vdash \forall  (A,d_A).s=t$ is derivable in the proof system of quantitative algebra by a finite proof. 
\item $\mathbb{ICA}\vdash \forall  (A,d_A).s=t$ is derivable in the proof system of quantitative algebra. 
\end{enumerate}
\end{corollary} 
\begin{proof}
Direction $(i)\Rightarrow(ii)$.  Following the same argument of Proposition \ref{example-UA-embedding}, the finite proof in Birkhoff's proof system of $s=t$
can be translated to a finite proof in the proof system of quantitative algebra of $\mathbb{ICA}\vdash \forall (A,d_{1}).s=t$,
where $d_1$ is the discrete fuzzy relation. The identity substitution $\sigma(a)=a$, of type $\sigma:(A,d_1)\rightarrow (A,d_A)$
is trivially provably 1-Lipschitz, because $d_1$ is discrete. Hence we can apply the substitution rule and obtain the desired judgment $\mathbb{ICA}\vdash \forall (A,d_A).s=t$.

Direction $(ii)\Rightarrow(iii)$ is trivial. 

Direction $(iii)\Rightarrow(i)$. Assume $\mathbb{ICA}\vdash \forall  (A,d_A).s=t$ is derivable.
By the soundness of the proof system, 
this means that, $\mathbb{ICA}\models \forall  (A,d_A).s=t$, i.e.~for any quantitative algebra $\mathbb{A}$ that is a model of $\mathbb{ICA}$, it holds that $\mathbb{A}\models  \forall  (A,d_A).s=t$.
Consider as $\mathbb{A}$ the quantitative algebra $(D(A),\{+_p\}_{p\in(0,1)}, K(d_A))$. Hence $\iota(s)=\iota(t)$ holds in $\mathbb{A}$
for all 1-Lipschitz interpretations $\iota:(A,d_A)\rightarrow (D(A),K(d_A))$. The intepretation $\eta_A(a)=\delta_a$ is 1-Lipschitz by Proposition \ref{unit_is_nonexpansive}.
Hence $\eta_A(s)=\eta_A(t)$. Since $(D(A), \{+_p\}_{p\in(0,1)})$ is the free convex algebra on $A$, this implies that $s=t$ is derivable in Birkhoff's proof system from the axioms of convex algebra.
\end{proof}


Note that the above Corollary (direction $(iii)\Rightarrow (ii)$) establishes half of what is required to conclude that $\mathbb{ICA}$ is compact.
We now proceed showing that the quantitative theory $\mathbb{ICA}$ and the
Kantorovich distance are also related by a completeness property (Proposition \ref{completeness:kantorovich} below). Before stating it, we prove
the following useful, purely syntactic,  lemma.

\begin{lemma}\label{technical_lemma_proofsystem}
Let $(A,d)$ be a (possibly infinite) fuzzy relation. Let $e(\vec{x})\in \textnormal{Terms}_{\Sigma_\textnormal{CA}}(\{x_1,\dots, x_n\} )$ be a convex algebra expression 
on $n$ variables of the form (using $n$-ary notation) 
$e(\vec{x})=\sum^n_{i=1} p_i x_i$. Let $\vec{a}=(a_1,\dots, a_n)$ and $\vec{b}=(b_1,\dots, b_n)$ be tuples in $A$. Denote with $e(\vec{a})$ and $e(\vec{b})$
the expressions (with variables in $A$) obtained by substituting $\vec{x}$ with $\vec{a}$ and $\vec{b}$.
Then the following judgment is derivable, by a finite proof, in the proof system of quantitative algebra:
\begin{center}
$
\mathbb{ICA} \vdash \forall (A,d). e(\vec{a}) =_\epsilon e(\vec{b}) \qquad\qquad \textnormal{ where }\epsilon =  \sum^n_{i=1} p_i\cdot d(a_i,b_i)\ .
$
\end{center}
\end{lemma}
\begin{proof}
By induction on $n$. 

Base case ($n=1$): in this case, $e(\vec{x})=x_1$ and we need to show a finite proof of the judgment:
$\mathbb{ICA} \vdash \forall (A,d). a_1 =_{d(a_1,b_1)} b_1$. And indeed it is simply derivable as an instance of one of the axioms of the proof system which 
uses the information provided by the distance $d$ on $A$.

Inductive Case ($n>1$): in this case $e(\vec{x}) = x_1 +_{p_1} e'(\vec{x})$ where the expression $e'(\vec{x})=\sum^n_{i=2} \frac{p_i}{1-p_1} x_i$ only involves variables $\{x_2,\dots, x_n\}$.
We need to exhibit a finite proof of the judgment 
\begin{center}
$
\mathbb{ICA} \vdash \forall (A,d). a_1 +_{p_1} e'(\vec{a}) =_{\epsilon} b_1 +_{p_1}e'(\vec{b}).
$
\end{center}
We derive this judgment by taking the following instance of the Interpolation axiom (in $\mathbb{ICA}$):
\begin{center}
  \begin{tabular}{l l}
    &  
    $\forall (
      \tiny{
      \begin{blockarray}{c c c c c}
      & x & y & w & z\\
      \begin{block}{c[cccc]}
        x & 1 & 1 & d(a_1,b_1) & 1 \\
        y & 1 & 1 & 1 & \epsilon' \\
        w & 1 & 1 & 1 & 1 \\
        z & 1 & 1 & 1 & 1 \\
      \end{block}
      \end{blockarray}
      }
      ). \ x +_{p_1} y\  =_{p_1d(a_1,b_1) + (1-p_1)\epsilon'} \ w+_{p_1} z$
  \end{tabular}
  \end{center}
where $\epsilon' = \sum^n_{i=2}\frac{p_i}{(1-p_1)}d(a_i,b_i) $, and applying the substitution $\tau:\{x,y,w,z\}\rightarrow \textnormal{Terms}_{\Sigma_\textnormal{CA}}(A)$ defined as: 
\begin{center}
$
\tau(x) = a_1 \qquad \tau(y) = e'(\vec{a}) \qquad \tau(w) = b_1  \qquad   \tau(z)= e'(\vec{b}).
$
\end{center}
This will yield the desired finite proof, since $\epsilon = p_1d(a_1,b_1) + (1-p_1)\epsilon'$. The application of the substitution rule requires showing that $\tau$ is provably 1-Lipschitz, i.e.,
we need to provide finite sub-derivations of the two judgments:
\begin{center}
$
\mathbb{ICA}\vdash \forall (A,d_A) . \tau(x) =_{d_{X}(x,w)}\tau(w) \qquad \qquad 
\mathbb{ICA}\vdash \forall (A,d_A) . \tau(y) =_{d_{X}(y,z)}\tau(z) 
$
\end{center}
where $d_X$ is the fuzzy relation on $\{x,y,w,z\}$ used in the instance of the Interpolation rule above.

For the left-judgment we need to derive $\forall (A,d) . a_1 =_{d(a_1, b_1)} b_1$. This, like in the base case, is derivable by instantiating one of the axioms of the proof system.
For the right-judgment we need a finite derivation  $\forall (A,d) . e'(\vec{a}) =_{\epsilon'} e'(\vec{b})$, and this exists by inductive hypothesis.
Hence $\tau$ is provably 1-Lipschitz and the substitution rule can be applied. 
\end{proof}



\begin{proposition}[Completeness]\label{completeness:kantorovich}
  Let $(A,d_A)$ be a fuzzy relation and $s,t\in\textnormal{Terms}_{\Sigma_{\textnormal{CA}}}(A)$ be two convex algebra terms with variables in $A$.
 Then the following implication holds,  for all $\epsilon\in[0,1]$:

\begin{center}
   $K(d)([s],[t])\leq \epsilon \qquad \Longrightarrow \qquad\mathbb{ICA}\vdash \forall (A,d).\ s =_\epsilon  t\ $ is derivable. 
\end{center}
\end{proposition}
\begin{proof}

Let us denote  $\mu=[s]$ and $\nu=[t]$. The proof involves two steps: 
\begin{enumerate}
  \item We first show that for every coupling $\gamma\in\Gamma(\mu,\nu)$ there exists a finite 
derivation of  the judgmet $\mathbb{ICA}\vdash \forall (A,d).\ s =_{\lambda}  t$, where $\lambda=V_d(\gamma)$ specified as in Definition \ref{Katorovich-defintion}.
\item Since $K(d)(\mu,\nu)\leq \epsilon$ by assumption, and $K(d)(\mu,\nu)=\inf_{\gamma}V_d(\gamma)$ by definition, we have that $\inf_\gamma V(\gamma)\leq \epsilon$. 
We can collect countably many derivations of (Step i) and, by applying the infinitary rule, obtain a derivation of:  
\begin{center}
$\mathbb{ICA}\vdash \forall (A,d).\ s =_{\delta}  t \qquad \qquad \delta=\inf_{\gamma}V_d(\gamma)$
\end{center}
and from this the desired derivation of 
  ($\mathbb{ICA}\vdash \forall (A,d).\ s =_{\epsilon}  t$)
is obtained by applying a weakening rule of the proof system, since $\epsilon \geq \delta$. 
\end{enumerate}
So it remains to outline the details of (Step i). Fix an arbitrary coupling $\gamma\in\Gamma(\mu,\nu)$.
By the defining properties of couplings, we have that, for all $a,b\in A$:
\begin{center}
$
\mu(a)= \sum_{b\in A} \gamma(a,b) \qquad \qquad \nu(b) = \sum_{a\in A} \gamma(a,b).
$
\end{center}
This implies that the convex algebra terms $s$ and $t$ can be proven equal (in the theory of convex algebras, using Birkhoff's proof system)
to the following convex algebra ($n$-ary) expressions:
\begin{center}
$
s = \displaystyle  \sum_{a\in A} (\sum_{b\in A} \gamma(a,b))\, a =  \displaystyle  \sum_{a,b\in A} \gamma(a,b)\, a   \qquad \qquad t = \displaystyle \sum_{b\in A}  ( \sum_{a\in A} \gamma(a,b)) \, b =  \displaystyle  \sum_{a,b\in A} \gamma(a,b)\, b .
$
\end{center}
In other words, there is a convex algebra expression $e(\vec{x})= \sum \gamma(a,b)\, x_{a,b}$, in the finite set of variables variables $\{ x_{a,b}\mid a,b\in supp(\mu)\cup supp(\nu)\}$,
such that the following equalities are provable from the axioms of convex algebras:
$s  = e(\vec{a})$ and $t = e(\vec{b})$, where $e(\vec{a})$ and $e(\vec{b})$ denote the substitutions of each of the variables $x_{a,b}$ in $e(\vec{x})$ with $\vec{a}$ and $\vec{b}$, respectively. 
By Corollary \ref{technical_corollary_proofsystem_extension}, these equalities have corresponding finite derivations of: 
$\mathbb{ICA}\vdash \forall (A,d). s =  e(\vec{a})$ and $\mathbb{ICA}\vdash \forall (A,d). t =  e(\vec{b}). $
From lemma \ref{technical_lemma_proofsystem} we obtain a finite derivation of the following judgment:
\begin{center}
$
\mathbb{ICA}\vdash \forall (A,d).\  e(\vec{a}) =_{\lambda} e(\vec{b}) \qquad \qquad \lambda = \displaystyle \sum_{a,b} \gamma(a,b)d(a,b)=V_d(\gamma)\ .
$
\end{center}
We can combine these three derivations, using a deductive rule of the proof system (Rule 4.j in \cite[Def 4.1]{MioEA24} called ``congruence''), to obtain the desired finite derivation of $\mathbb{ICA}\vdash \forall (A,d).\  s =_{\lambda} t$.
\end{proof}

\begin{corollary}
The free $\mathbb{ICA}$-quantitative algebra generated by the fuzzy relation $(A,d_A)$ is isomorphic to $(D(A),\{+_p\}_{p\in(0,1)},K(d_A))$.
\end{corollary}
\begin{proof}
As discussed at the end of Section \ref{sec:background:qa}, the carrier of the free quantitative algebra is $\textnormal{Terms}_{\Sigma_{\textnormal{CA}}}(A)/_{\equiv}$ by the relation $\equiv$ of provable equality, and this coincides (by Corollary \ref{technical_corollary_proofsystem_extension}),
 with provable equality in Birkhoff's proof system from the axioms of convex algebras. Hence the carrier is $(D(A) ,\{+_p\}_{p\in(0,1)})$, the free convex algebra on $A$.
The distance $d_F$ of the free quantitative algebra is defined as provable distance and, by the soundness (\ref{soundness:proposition1}) and completeness (\ref{completeness:kantorovich}) results, it coincides with $K(d_A)$.
\end{proof}

We are finally ready to show that the quantitative theory $\mathbb{ICA}$ is compact. 
The statement of completeness (Proposition \ref{completeness:kantorovich}) does not assert the existence
of a finite proof. By inspection of its proof (specifically Step (ii)), this is due to one single application of the infinitary rule.
But the following well known Proposition \ref{existed_optimal_coupling}, stating that there always exists an \emph{optimal} coupling $\gamma$ achieving the infimium in the definition of $K(d)$, 
ensures that Step (ii) is unnecessary. Hence the statement of the completeness Proposition \ref{completeness:kantorovich} can be strengthen 
by stating that a \emph{finite} proof exists (Corollary \ref{completeness:kantorovich:finite} below). We have delayed this simple observation, and its consequence, to highlight that the soundness and completeness results do not depend on 
the assumption that optimal couplings exist.

\begin{proposition}\label{existed_optimal_coupling}
Let $(A,d)$ be  a fuzzy relation. Then, for every $\mu,\nu\in D(A)$ there exists an optimal coupling $\gamma_\star \in\Gamma(\mu,\nu)$ achieving the infimum: $V_d(\gamma_\star)=K(d)(\mu,\nu)$.
\end{proposition}
\begin{proof}
  By Proposition \ref{Heine-Borel-lemmma}, $\Gamma(\mu,\nu)$ is a compact topological space and $V_d:\Gamma(\mu,\nu)\rightarrow [0,1]$ is a continuous function and, as such, it attains a minimum value (see \cite[Thm 4, p. 362]{bourbaki_book}).
\end{proof}
\begin{corollary}[Completeness wrt finite proofs]\label{completeness:kantorovich:finite}
  Let $(A,d)$ be a fuzzy relation and $s,t\in\textnormal{Terms}_{\Sigma_{\textnormal{CA}}}(A)$ be two convex algebra terms with variables in $A$.
  Then, for all $\epsilon\in[0,1]$:
\begin{center}
   $K(d)([s],[t])\leq \epsilon \qquad \Longrightarrow \qquad\mathbb{ICA}\vdash \forall (A,d).\ s =_\epsilon  t\ $ is derivable with a finite proof. 
\end{center}
\end{corollary}

\begin{theorem}\label{theorem:ica:compact}
The quantitative equational theory $\mathbb{ICA}$ is compact.
\end{theorem}
\begin{proof}
Let $\phi$ be a quantitative equation and assume $\mathbb{ICA}\models \phi$ or, equivalently,
that $\mathbb{ICA}\vdash \phi$ is derivable by a possibly infinite proof. 
First, if $\phi$ is of the form $\forall (A,d).s=t$, by Proposition \ref{technical_corollary_proofsystem_extension} (iii)$\Rightarrow$(ii),  
$\mathbb{ICA}\vdash \phi$ has a finite proof. Second, if $\phi$ is of the form $\forall (A,d).s=_\epsilon t$, By Proposition \ref{soundness:proposition1} we have that 
$K(d)([s],[t])\leq \epsilon$ and from Proposition \ref{completeness:kantorovich:finite} we have a finite proof of $\mathbb{ICA}\vdash \phi$.
\end{proof}

\section{Analysis of the Proof and Generalizations}\label{section:generalisation}
In this section we analyze the proof of compactness of $\mathbb{ICA}$ (Theorem \ref{theorem:ica:compact})
to get general insights and generalzations. As a result we obtain a family of compact quantitative theories of convex algebras (Theorem \ref{generalisation:theorem:new} below)
including some interesting examples.

First, by Definition \ref{definition_thery_ICA}, the theory $\mathbb{ICA}$ includes quantitative equations corresponding to the equations of convex algebras
(``Idempotency'', ``Skew-comm./assoc.'' ) and quantitative equations of the form:
\begin{center}
  \begin{tabular}{l l}
    Interpolative: &  $\forall (
      \tiny{
      \begin{blockarray}{c c c c c}
      & x & y & w & z\\
      \begin{block}{c[cccc]}
        x & 1 & 1 & \epsilon & 1 \\
        y & 1 & 1 & 1 & \delta \\
        w & 1 & 1 & 1 & 1 \\
        z & 1 & 1 & 1 & 1 \\
      \end{block}
      \end{blockarray}
      }
      ). \ x +_p y\  =_{\epsilon \oplus_p \delta} \ w+_{p} z$.
  \end{tabular}
  \end{center}
where $\oplus_p:[0,1]^2\rightarrow [0,1]$ is the standard convex algebra operation on $[0,1]$, defined as $x\oplus_p y = px +(1-p)y$.
We highlight with $\oplus_p$ the standard operation because, as we discuss later, it is by varying $\oplus_p$  
that we will obtain other examples of compact theories.

Second, a lifting method (the Kantorovich distance) is found to map a fuzzy relation $(A,d)$  
to a fuzzy relation on $D(A)$, which is the free convex algebra on $A$.
Crucially, the lifting method is expressed as an infimum ($\inf_\gamma V_d(\gamma)$) over evaluations (by the function $V_d:\Gamma(\mu,\nu)\rightarrow [0,1]$) of 
a family of ``witness objects'' (the set $\Gamma(\mu,\nu)$ of couplings $\gamma$).

Third, the soundness of this lifting method is established (Proposition \ref{soundness:proposition1} and Corollary \ref{soundness:corollary}).
The key facts used in the proof of soundness are:
\begin{enumerate}
  \item the witness objects (couplings) carry a convex algebra structure, as in Lemma \ref{lemma:soundness}. 
  \item the evaluation $V_d$ satisfies the property: $V_d(\delta_{\langle a,b \rangle}) \leq d(a,b)$, to validate Proposition \ref{unit_is_nonexpansive}.
  \item the proof of Proposition  \ref{soundness:proposition1} requires that $V_d$ satisfies: $V_d(\gamma +_p \gamma') \leq V_d(\gamma) \oplus_p V_d(\gamma')$. 
  \item the proof of Proposition  \ref{soundness:proposition1} requires that the convex algebra operation $\oplus_p:[0,1]^2\rightarrow [0,1]$ is  monotone (with respect to the standard order on $[0,1]$) in both arguments.
   Monotonicity is required to derive, from the assumptions $V_d(\gamma)\leq \epsilon +\lambda$ and $V_d(\gamma')\leq \delta +\lambda$, that 
   $V_d(\gamma) \oplus_p V_d(\gamma') \leq (\epsilon +\lambda) \oplus_p(\delta +\lambda) $.
   \item the proof of Proposition  \ref{soundness:proposition1} requires that the convex algebra operation $\oplus_p:[0,1]^2\rightarrow [0,1]$ satisfies: $\lim_{\lambda \rightarrow 0}  (\epsilon +\lambda) \oplus_p(\delta +\lambda)  \leq  \epsilon  \oplus_p \delta$. 
   This property is used to derive  $V_d(\gamma) \oplus_p V_d(\gamma')\leq  \epsilon  \oplus_p \delta$.
\end{enumerate}

Fourth, the completeness of the lifting method is established (Proposition \ref{completeness:kantorovich}).
The key facts used in the proof of completeness are:
\begin{enumerate}
\item The proof is based on Lemma \ref{technical_lemma_proofsystem}. This lamma is purely syntactic and 
 it states that, given two convex algebra terms $e(\vec{a})$ and $e(\vec{b})$ of the same shape,  it is possible to iterate the ``Interpolative'' axioms to derive a bound $\epsilon$ on their distances.
Specifically, if $e(\vec{x})=\sum_i p_i x_i$ then $\epsilon = \bigoplus_i p_i d(a_i,b_i)$. 
\item The proof of completeness is based on the observation that a coupling $\gamma\in\Gamma(\mu,\nu)$ provides
a ``receipe'' for rewriting, using the axioms of convex algebras, the two distributions $\mu$ and $\nu$ in the same shape (the convex algebra term $e(\vec{x})$)
which allows for the application of Lemma \ref{technical_lemma_proofsystem}.
\end{enumerate}

Fifth, and last point, the proof of compactness of Theorem \ref{theorem:ica:compact} uses the finitary completeness Corollary \ref{completeness:kantorovich:finite} which is a consequence of the fact that $\inf_\gamma V_d(\gamma)$ achieves a minimum.
This fact is, in turn, proved using the topological compactness of $\Gamma(\mu,\nu)$, from Lemma \ref{Heine-Borel-lemmma},
and topological properties of $V_d$ (continuity) ensuring that a minimum is always achived on a compact domain.

The analysis suggests a generalization of the results of Section \ref{main:section:result} parametric on choices of $\oplus_p$ different than
the standard one. In the rest of this section we develop this idea. The main result, stated as Theorem \ref{generalisation:theorem:new},
gives a family of examples of compact quantitative theories, including the axiomatisation of the $k$-Wasserstein distance, for $k\geq 1$, which
was also studied in \cite{DBLP:conf/lics/MardarePP16}, and the $\infty$-Wasserstein distance.

First, we generalize the notion of Kantorovich distance (Definition \ref{Katorovich-defintion})
and of $\mathbb{ICA}$ (Definition \ref{definition_thery_ICA}) to arbitrary convex algebras  $([0,1], \{\oplus_p\}_{p\in (0,1)})$ on $[0,1]$.

\begin{definition}[Generalized lifting $K^{\oplus_p}(d)$]
  Let $(A,d)$ be a fuzzy relation, i.e., $d:A\times A\rightarrow [0,1]$.
  Denote with $V^{\oplus_p}_{d}$ the unique (recall that $D(A\times A)$ is free on $A\times A$) homomorphic extension of $d$, defined as:
  \begin{center}
  $
  V^{\oplus_p}_{d}\colon D(A\times A)\rightarrow [0,1] \qquad \qquad\qquad 
   \displaystyle \sum^n_{i=i} p_i \delta_{\langle a_i,b_i\rangle} \mapsto \displaystyle \bigoplus^n_{i=1} p_i d(a_i,b_i)\ .
  $
  \end{center} 
  The \emph{$\oplus_p$-lifting of $d$} is the fuzzy relation $K^{\oplus_p}(d)$ on $D(A)$ defined as follows, for all $\mu,\nu\in D(A)$:
  \begin{center}
  $
  K^{\oplus_p}(d)(\mu,\nu) = \displaystyle \inf_{\gamma\in\Gamma(\mu,\nu)} V^{\oplus_p}_{d}(\gamma)\ . 
  $
  \end{center}
\end{definition}
\begin{definition}[Generalized quantitative theory $\mathbb{ICA}^{\oplus_p}$]
The quantitative equational theory  $\mathbb{ICA}^{\oplus_p}$ is defined by the union of the quantitative equations 
  ``Idempotency'', ``Skew-comm'' and ``Skew-assoc'' of Definition \ref{definition_thery_ICA} and by the ``$\oplus_p$-Interpolative''
  quantitative equations, for all $p\in(0,1)$ and $\epsilon,\delta\in[0,1]$:
  \begin{center}
    \begin{tabular}{l l}
      $\oplus_p$--Interpolative: &  $\forall (
        \tiny{
        \begin{blockarray}{c c c c c}
        & x & y & w & z\\
        \begin{block}{c[cccc]}
          x & 1 & 1 & \epsilon & 1 \\
          y & 1 & 1 & 1 & \delta \\
          w & 1 & 1 & 1 & 1 \\
          z & 1 & 1 & 1 & 1 \\
        \end{block}
        \end{blockarray}
        }
        ). \ x +_p y\  =_{\epsilon \oplus_p \delta} \ w+_{p} z$ .
    \end{tabular}
    \end{center}
\end{definition}
Hence, $K(d)$ and $\mathbb{ICA}$ are instances of $K(d)^{\oplus_p}$ and $\mathbb{ICA}^{\oplus_p}$ for the standard 
convex algebra $([0,1], \{\oplus_p\}_{p\in (0,1)})$ on $[0,1]$, where $x\oplus_p y = px + (1-p)y$.
\begin{theorem}\label{generalisation:theorem:new}
  Let  $([0,1], \{\oplus_p\}_{p\in (0,1)})$ be a convex algebra on $[0,1]$ such that, for all $p\in (0,1)$:
\begin{enumerate}
  \item $\oplus_p$ is monotone: $\forall x,x',y\in [0,1]$, if $x\leq x'$ then $x\oplus_p y\leq x'\oplus_p y$ and $y\oplus_p x\leq y\oplus_p x'$,
  \item for all $x,y\in[0,1]$, $\displaystyle \lim_{\lambda\rightarrow 0}\big( (x+\lambda) \oplus_p  (y+\lambda)\big) \leq x\oplus_p y$,
\end{enumerate}
then $\mathbb{ICA}^{\oplus_p}$ is a sound and complete axiomatization of the distance lifting $K^{\oplus_p}$. Furthermore, if:
\begin{enumerate}
  \setcounter{enumi}{2}
  \item for all $d:A\times A\rightarrow [0,1]$ and $\mu,\nu\in D(A)$,  there exists an optimal coupling:
  \begin{center}
$\exists \gamma_\star\in \Gamma(\mu,\nu). \qquad\qquad 
 V^{\oplus_p}_{d}(\gamma_\star) \ = \ 
  \displaystyle \inf_{\gamma\in\Gamma(\mu,\nu)} V^{\oplus_p}_{d}(\gamma)$
  \end{center}
\end{enumerate}
then $\mathbb{ICA}^{\oplus_p}$ is a compact quantitative equational theory.
\end{theorem}
\begin{proof}
By definition of $V^{\oplus_p}_{d}$ as a homomorphism of convex algebras,  it holds that $V^{\oplus_p}_{d}(\delta_{\langle a,b\rangle}) = d(a,b)$ and
$V^{\oplus_p}_{d}(\gamma+_q \gamma') =  V^{\oplus_p}_{d}(\gamma)\oplus_q V^{\oplus_p}_{d}(\gamma')$. This, together with the assumptions (i) and (ii),
ensures that the proof of Soundness can be carried out. Similarly, using the definition $\mathbb{ICA}^{\oplus_p}$,
the analogous of Lemma \ref{technical_lemma_proofsystem} can be proved and with it the proof of completeness. 
Finally, the existence of optimal couplings allows for the strenthening of the completeness theorem to finite proofs, ensuring compactness. 
\end{proof}
For a given convex algebra $([0,1], \{\oplus_p\}_{p\in (0,1)})$ on $[0,1]$, conditions (i) and (ii) are usually straightforwards to verify. For (iii), it is typically necessary to establish some topological 
property (like continuity or, more generally, lower semicontinuity \cite[\S 6]{bourbaki_book}) of the evaluation map $V^{\oplus_p}_{d}:D(A\times A)\rightarrow [0,1]$ that guarantees that the infimum over the compact set $\Gamma(\mu,\nu)$ is always achieved (see \cite[Thm 4, p. 362]{bourbaki_book}).

\begin{proposition}[Examples]\label{proposition:examples:final}
  The following operations $\oplus_p:[0,1]^2\rightarrow [0,1]$, for $p\in (0,1)$, satisfy the axioms of convex algebras (Definition \ref{convex-algebras-definition}) and conditions (i--iii) of Theorem \ref{generalisation:theorem:new}:
 \begin{center}
\begin{tabular}{l l l }
binary version & $n$-ary version \\
\hline
\hline
$x\oplus_p y = px +(1-p)y$ & $\bigoplus_i p_i x_i = \sum_i p_ix_i$ & the standard operation\\
\hline
$x\oplus_p y = \max\{ x,y\}$ & $\bigoplus_i p_i x_i = \max\{ x_i \}_i$ & $\vee$-semilattice\\
\hline
$x\oplus_p y = (px^k + (1-p)y^k)^{\frac{1}{k}}$ & $\bigoplus_i p_i x_i = (\sum_i p_i x_i^k)^{\frac{1}{k}}$ & for some $k\geq 1$\\
\hline
$x\oplus_p y = x^p\cdot y^{1-p}$ &  $\bigoplus_i p_i x_i = \prod_i {x_i}^{p_i}$ & log-probabilities\\
\hline
\end{tabular}
 \end{center}
\end{proposition}

Hence, by Theorem \ref{generalisation:theorem:new}, the quantitative theory $\mathbb{ICA}^{\oplus_p}$ is compact for all these examples. As already noted, the Kantorovich distance $K(d)$ (Definition \ref{Katorovich-defintion}) is the special case of $K^{\oplus_p}(d)$ for $x\oplus_p y = px + (1-p)y$, the standard convex algebra operation on $[0,1]$.
The case $x \oplus_p y = (px^k + (1-p)y^k)^{\frac{1}{k}}$, for some $k\geq 1$, gives as $K^{\oplus_p}(d)$ the Wasserstein $k$-distance (see \cite[Ch. 6]{villani2008}).
The case $x \oplus_p y= \max\{x,y\}$ gives as $K^{\oplus_p}(d)$ the Wasserstein $\infty$-distance (see \cite{infwasserstain}). Note that, in this case, the evaluation map $V^{\oplus_p}_{d}$ is not continuous
but it is lower semicontinuous by \cite[Thm. 4, p. 362]{bourbaki_book}, since it is the supremum of all $k$-Wasserstein evaluations, for $k\geq 1$, which are continuous.
Finally, the case $x \oplus_p y= x^py^{1-p}$ is a distance related to log-probabilities, since $x^py^{1-p} = \exp(p\log(x) + (1-p)\log(y))$, using the convention $\log(0)=-\infty$ and $\exp(-\infty)=0$.

\section{Conclusions and Future Work}
We have introduced the concept of \emph{compact quantitative equational theory} in the quantitative algebra apparatus of \cite{MioEA24} and proved that the theory $\mathbb{ICA}$ of \emph{interpolative convex algebras} is compact. This theory was first studied in \cite{DBLP:conf/lics/MardarePP16},
under the name \emph{interpolative barycentric algebras}, in the less general context of metric spaces and wihtout explictly observing the compactness of the theory. 
In Section \ref{section:generalisation} we have obtained an algebraic generalization: to each convex algebra structure on $[0,1]$,
satisfying certain regularity conditions, corresponds 
a compact variant of $\mathbb{ICA}$ axiomatizing a variant of the Kantorovich lifting on probability distributions. 
We have discussed how some of these variants coincide with well-known distance liftings, such as $k$-Wasserstein (also studied in \cite{DBLP:conf/lics/MardarePP16}), $\infty$-Wasserstein and involving $\log$-probabilities. 

Methods similar to those adopted in this paper can be used to prove the compactness of the quantitative theories of semilattices \cite{DBLP:conf/lics/MardarePP16} and of convex  semilattices \cite{DBLP:conf/concur/MioV20}. 
Both these theories, just like $\mathbb{ICA}$, turn out to be manageble due to topological compactness reasons. For example, finitely generated convex sets of finitely supported probability distributions are compact.
An interesting direction of future research is to explore the connection between compactness of quantitative theories and topological compactness.

Our definitions and results have been formulated in the context of the quantitative algebra apparatus of \cite{MioEA24} based on fuzzy relations, but
they can be adapted to the same apparatus where \emph{generalized metric spaces} (such as metric spaces) take the place of fuzzy relations, as described in \cite[\S 8]{MioEA24},
or in the quantitative algebra apparatus of \cite{FordEA21} based on arbitrary relational Horn theories.


\bibliographystyle{entics}
\bibliography{biblio}

\end{document}